\documentclass[12pt]{elsarticle} 
\usepackage[margin=2.5cm]{geometry}% by courtesy of Mico
\usepackage{lipsum}

\makeatletter
\def\ps@pprintTitle{%
   \let\@oddhead\@empty
   \let\@evenhead\@empty
   \def\@oddfoot{\reset@font\hfil\thepage\hfil}
   \let\@evenfoot\@oddfoot
}
\makeatother
\usepackage{amsfonts,color,morefloats,pslatex}
\usepackage{amssymb,amsthm, amsmath,latexsym}

\newtheorem{theorem}{Theorem}
\newtheorem{lemma}[theorem]{Lemma}

\newtheorem{corollary}[theorem]{Corollary}

\newtheorem{open}[theorem]{Open Problem}

\newtheorem{example}[theorem]{Example}

\newtheorem{conj}[theorem]{Conjecture}

\newcommand\myatop[2]{\genfrac{}{}{0pt}{}{#1}{#2}}

\newcommand{\tr}{{\mathrm{Tr}}}

\newcommand{\gf}{{\mathrm{GF}}}
\newcommand{\PG}{{\mathrm{PG}}}
\newcommand{\AG}{{\mathrm{AG}}}

\newcommand{\support}{{\mathrm{suppt}}}
\newcommand{\Aut}{{\mathrm{Aut}}} 
\newcommand{\PAut}{{\mathrm{PAut}}} 
\newcommand{\MAut}{{\mathrm{MAut}}} 
\newcommand{\GAut}{{\mathrm{Aut}}}

\newcommand{\Sym}{{\mathrm{Sym}}}

\newcommand{\GL}{{\mathrm{GL}}}

\newcommand{\wt}{{\mathtt{wt}}}

\newcommand{\cP}{{\mathcal{P}}} 
\newcommand{\cB}{{\mathcal{B}}} 
 
\newcommand{\cI}{{\mathcal{I}}} 
\newcommand{\C}{{\mathsf{C}}} 
\newcommand{\M}{{\mathsf{M}}}

\newcommand{\cR}{{\mathcal{R}}}

\newcommand{\bc}{{\mathbf{c}}} 
\newcommand{\bg}{{\mathbf{g}}} 
\newcommand{\bm}{{\mathbf{m}}}

\newcommand{\bzero}{{\mathbf{\bar{0}}}}
\newcommand{\bone}{{\mathbf{\bar{1}}}}

\newcommand{\bD}{{\mathbb{D}}}

\newcommand{\GA}{{\mathrm{GA}}}

\usepackage{blindtext}

\begin{document}
%\tableofcontents

\begin{frontmatter}

%% Title, authors and addresses

%% use the tnoteref command within \title for footnotes;
%% use the tnotetext command for the associated footnote;
%% use the fnref command within \author or \address for footnotes;
%% use the fntext command for the associated footnote;
%% use the corref command within \author for corresponding author footnotes;
%% use the cortext command for the associated footnote;
%% use the ead command for the email address,
%% and the form \ead[url] for the home page:
%%
%% \title{Title\tnoteref{label1}}
%% \tnotetext[label1]{}
%% \author{Name\corref{cor1}\fnref{label2}}
%% \ead{email address}
%% \ead[url]{home page}
%% \fntext[label2]{}
%% \cortext[cor1]{}
%% \address{Address\fnref{label3}}
%% \fntext[label3]{}

\title{The linear codes of $t$-designs held in the Reed-Muller and Simplex codes       
\tnotetext[fn1]{C. Ding's research was supported by the Hong Kong Research Grants Council,
Proj. No. 16300418. C. Tang was supported by National Natural Science Foundation of China (Grant No.
11871058) and China West Normal University (14E013, CXTD2014-4 and the Meritocracy Research
Funds).}
}

\author[cding]{Cunsheng Ding}
\ead{cding@ust.hk}
\author[cmt]{Chunming Tang}
\ead{tangchunmingmath@163.com}

%\cortext[lcj]{Corresponding author}
\address[cding]{Department of Computer Science and Engineering, The Hong Kong University of Science and Technology, Clear Water Bay, Kowloon, Hong Kong, China}
\address[cmt]{School of Mathematics and Information, China West Normal University, Nanchong, Sichuan,  637002, China}
%\address[vdt]{
%Department of Mathematical Sciences, Michigan Technological University, 
%Houghton, Michigan 49931, USA
%}

%% use optional labels to link authors explicitly to addresses:
%% \author[label1,label2]{<author name>}
%% \address[label1]{<address>}
%% \address[label2]{<address>}
%\author{Cunsheng Ding}
%\ead{cding@ust.hk} 

%\cortext[lcj]{Corresponding author}
%\address{Department of Computer Science and Engineering, 
%The Hong Kong University of Science and Technology,
%Clear Water Bay, Kowloon, Hong Kong, China}

%\tableofcontents

\begin{abstract}
A fascinating topic of combinatorics is $t$-designs, which have a very long history. 
The incidence matrix of a $t$-design generates a linear code over $\gf(q)$ for any prime power 
$q$, which is called the linear code of the $t$-design over $\gf(q)$. On the other hand, some linear codes 
hold $t$-designs for some $t \geq 1$. The purpose of this paper is to study the linear codes of some 
$t$-designs held in the Reed-Muller and Simplex codes. Some general theory for the linear codes of $t$-designs held in 
linear codes is presented. Open problems are also presented.         
\end{abstract}

\begin{keyword}
Cyclc code \sep linear code \sep Reed-Muller code \sep $t$-design.
%% PACS codes here, in the form: \PACS code \sep code

%% MSC codes here, in the form: \MSC code \sep code
%% or \MSC[2008] code \sep code (2000 is the default)
\MSC  05B05 \sep 51E10 \sep 94B15 

\end{keyword}

\end{frontmatter}

%\tableofcontents

\section{Introduction}

Let $\cP$ be a set of $v \ge 1$ elements, and let $\cB$ be a set of $k$-subsets of $\cP$, where $k$ is
a positive integer with $1 \leq k \leq v$. Let $t$ be a positive integer with $t \leq k$. The pair
$\bD = (\cP, \cB)$ is called a $t$-$(v, k, \lambda)$ {\em design\index{design}}, or simply {\em $t$-design\index{$t$-design}}, if every $t$-subset of $\cP$ is contained in exactly $\lambda$ elements of
$\cB$. The elements of $\cP$ are called points, and those of $\cB$ are referred to as blocks.
We usually use $b$ to denote the number of blocks in $\cB$.  A $t$-design is called {\em simple\index{simple}} if $\cB$ does not contain repeated blocks. In this paper, we consider only simple 
$t$-designs.  A $t$-design is called {\em symmetric\index{symmetric design}} if $v = b$. It is clear that $t$-designs with $k = t$ or $k = v$ always exist. Such $t$-designs are {\em trivial}. In this paper, we consider only $t$-designs with $v > k > t$.
A $t$-$(v,k,\lambda)$ design is referred to as a {\em Steiner system\index{Steiner system}} if $t \geq 2$ and $\lambda=1$, and is denoted by $S(t,k, v)$. 

\subsection{The codes of designs}

Let $\bD = (\cP, \cB)$ be a $t$-$(v, k, \lambda)$ design with $b \ge 1$ blocks. 
The points of $\cP$ are usually indexed with $p_1,p_2,\cdots,p_v$, and the blocks of $\cB$ 
are normally denoted by $B_1, B_2, \cdots, B_b$. The {\em incidence matrix\index{incidence 
matrix}} $M_\bD=(m_{ij})$ of $\bD$ is a $b \times v$ matrix where $m_{ij}=1$ if  $p_j$ is on $B_i$ 
and $m_{ij}=0$ otherwise. The binary matrix $M_{\bD}$ is viewed as a matrix over $\gf(q)$ for any 
prime power $q$, and its row vectors span a linear code of length $v$ over $\gf(q)$, which is 
denoted by $\C_q(\bD)$ and called the \emph{code} of $\bD$ over $\gf(q)$. 
It is clear that the code $\C_q(\bD)$ depends on the labelling of the points 
of $\bD$, but is unique up to coordinate permutations.

\subsection{The support designs of linear codes}

We assume that the reader is familiar with the basics of linear codes and cyclic codes, 
and proceed to introduce the support designs of linear codes directly. 
Let $\C$ be a $[v, \kappa, d]$ linear code over $\gf(q)$. Let $A_i:=A_i(\C)$, which denotes the
number of codewords with Hamming weight $i$ in $\C$, where $0 \leq i \leq v$. The sequence 
$(A_0, A_1, \cdots, A_{v})$ is
called the \textit{weight distribution} of $\C$, and $\sum_{i=0}^v A_iz^i$ is referred to as
the \textit{weight enumerator} of $\C$. For each $k$ with $A_k \neq 0$,  let $\cB_k$ denote
the set of the supports of all codewords with Hamming weight $k$ in $\C$, where the coordinates of a codeword
are indexed by $(p_1, \ldots, p_v)$. Let $\cP=\{p_1, \ldots, p_v\}$.  The pair $(\cP, \cB_k)$
may be a $t$-$(v, k, \lambda)$ design for some positive integer $\lambda$, which is called a 
\emph{support design} of the code, and is denoted by $\bD_k(\C)$. In such a case, we say that the code $\C$ holds a $t$-$(v, k, \lambda)$ 
design. Throughout this paper, we denote the dual code of $\C$ by $\C^\perp$, and the extended code of 
$\C$ by $\overline{\C}$.

\subsection{The objectives of this paper}

While most linear codes over finite fields do not hold $t$-designs, some linear codes do hold 
$t$-designs for $t \geq 1$. Studying the linear codes of $t$-designs has been a topic of research for a long time 
\cite{AK92,AK98,AM69,CH92,Dingbk15,HP03,KP95,Tonc93,Tonchev,Tonchevhb}. 

Let $q_1$ be a power of a prime $p$. Our starting point is a linear code $\C_1$ over a finite field $\gf(q_1)$, which holds a $t$-design $\bD_k(\C_1)$, 
our objective is to study the classical linear code $\C_2=\C_{q_2}(\bD_k(\C_1))$ over a finite field $\gf(q_2)$, and 
hope that the new code $\C_2$ has interesting parameters and properties. This idea is depicted as follows: 
\begin{eqnarray*}
\mbox{Original code $\C_1$ over $\gf(q_1)$}  \Rightarrow 
\mbox{A $t$-design $\bD_k(\C_1)$ held in $\C_1$} \Rightarrow 
\mbox{New code $\C_2:=\C_{q_2}(\bD_k(\C_1))$}.  
\end{eqnarray*} 
It may happen that $\C_2=\C_1$, but they are different in many cases. Note that a linear code $\C_1$ may 
hold exponentially many $t$-designs. We may obtain exponentially many new codes $\C_2=\C_{q_2}(\bD_k(\C_1))$ 
from the original code $\C_1$. Although the finite field $\gf(q_2)$ has many choices for the given $\C_1$ and 
$q_1$, we will restrict ourself to the case $q_2=p$ for simplicity in most parts of this paper. 
It is well known 
that the code  $\C_{p}(\bD)$ of a $t$-$(v, k, \lambda)$ design $\bD$ has dimension less than $v-1$ only if $p$ 
divides $\lambda_1 - \lambda_2$, where $\lambda_i$
denotes the number of blocks of $\bD$ that contain $i$ points
($i=1, 2$) (cf.  \cite{Hamada73}, \cite[Theorem 1.86]{Tonchevhb}.)

In this paper, we will consider several families of linear codes $\C$ over $\gf(q)$ and some of the designs $\bD_k(\C)$ 
held in $\C$, and will determine the parameters of the linear code $\C_p(\bD_k(\C))$ for some designs $\bD_k(\C)$ held in 
$\C$. This is doable in the case that $q=2$, but is a hard problem for $q>2$. In the binary case, we will present 
some general theory for codes   $\C_2(\bD_k(\C))$.  
The objective of this paper is to study the linear codes of some known 
$t$-designs held in the generalised Reed-Muller codes and the Simplex codes. Some general theory for the linear codes of $t$-designs held in linear codes is presented.   Open problems on this topic will also be presented.

\section{Auxiliary results}

\subsection{Designs from linear codes via the Assmus-Mattson Theorem}

The following theorem, developed by Assumus and Mattson, shows that the pair $(\cP, \cB_k)$ defined by 
a linear code is a $t$-design under certain conditions.

\begin{theorem}[Assmus-Mattson Theorem]\label{thm-designAMtheorem} (\cite{AM69}, \cite[p. 303]{HP03})
Let $\C$ be a $[v,k,d]$ code over $\gf(q)$. Let $d^\perp$ denote the minimum distance of $\C^\perp$. 
Let $w$ be the largest integer satisfying $w \leq v$ and 
$$ 
w-\left\lfloor  \frac{w+q-2}{q-1} \right\rfloor <d. 
$$ 
Define $w^\perp$ analogously using $d^\perp$. Let $(A_i)_{i=0}^v$ and $(A_i^\perp)_{i=0}^v$ denote 
the weight distribution of $\C$ and $\C^\perp$, respectively. Fix a positive integer $t$ with $t<d$, and 
let $s$ be the number of $i$ with $A_i^\perp \neq 0$ for $1 \leq i \leq v-t$. Suppose $s \leq d-t$. Then 
\begin{itemize}
\item the codewords of weight $i$ in $\C$ hold a $t$-design provided $A_i \neq 0$ and $d \leq i \leq w$, and 
\item the codewords of weight $i$ in $\C^\perp$ hold a $t$-design provided $A_i^\perp \neq 0$ and 
         $d^\perp \leq i \leq \min\{v-t, w^\perp\}$. 
\end{itemize}
\end{theorem}

The Assmus-Mattson Theorem is a very useful tool for constructing $t$-designs from linear codes, 
and has been recently employed to construct infinitely many $2$-designs and $3$-designs.

\subsection{Designs from linear codes via the automorphism group}

In this section, we introduce the automorphism approach to obtaining $t$-designs from linear codes. 
To this end, we have to define the automorphism group of linear codes. We will also present some 
basic results about this approach.

The set of coordinate permutations that map a code $\C$ to itself forms a group, which is referred to as 
the \emph{permutation automorphism group\index{permutation automorphism group of codes}} of $\C$
and denoted by $\PAut(\C)$. If $\C$ is a code of length $n$, then $\PAut(\C)$ is a subgroup of the 
\emph{symmetric group\index{symmetric group}} $\Sym_n$.

A \emph{monomial matrix\index{monomial matrix}} over $\gf(q)$ is a square matrix having exactly one 
nonzero element of $\gf(q)$  in each row and column. A monomial matrix $M$ can be written either in 
the form $DP$ or the form $PD_1$, where $D$ and $D_1$ are diagonal matrices and $P$ is a permutation 
matrix. 

The set of monomial matrices that map $\C$ to itself forms the group $\MAut(\C)$,  which is called the 
\emph{monomial automorphism group\index{monomial automorphism group}} of $\C$. Clearly, we have 
$$
\PAut(\C) \subseteq \MAut(\C).
$$

The \textit{automorphism group}\index{automorphism group} of $\C$, denoted by $\GAut(\C)$, is the set 
of maps of the form $M\gamma$, 
where $M$ is a monomial matrix and $\gamma$ is a field automorphism, that map $\C$ to itself. In the binary 
case, $\PAut(\C)$,  $\MAut(\C)$ and $\GAut(\C)$ are the same. If $q$ is a prime, $\MAut(\C)$ and 
$\GAut(\C)$ are identical. In general, we have 
$$ 
\PAut(\C) \subseteq \MAut(\C) \subseteq \GAut(\C). 
$$

By definition, every element in $\GAut(\C)$ is of the form $DP\gamma$, where $D$ is a diagonal matrix, 
$P$ is a permutation matrix, and $\gamma$ is an automorphism of $\gf(q)$.   
The automorphism group $\GAut(\C)$ is said to be $t$-transitive if for every pair of $t$-element ordered 
sets of coordinates, there is an element $DP\gamma$ of the automorphism group $\GAut(\C)$ such that its 
permutation part $P$ sends the first set to the second set.

The next theorem gives another sufficient condition for a linear code to hold $t$-designs \cite[p. 308]{HP03}. 
 
\begin{theorem}\label{thm-designCodeAutm}
Let $\C$ be a linear code of length $n$ over $\gf(q)$ where $\GAut(\C)$ is $t$-transitive. Then the codewords of any weight $i \geq t$ of $\C$ hold a $t$-design.
\end{theorem}

\subsection{Relations between $\C_q(\bD)$ and  $\C_q(\bD^c)$} 

Let $\bD$ be a $t$-$(v, k, \lambda)$ design. Then its complement $\bD^c$ is a $t$-$(v, v-k, \lambda^c)$ design, where 
$$ 
\lambda^c=\lambda \frac{\binom{v-t}{k}}{\binom{v-t}{k-t}}. 
$$ 
Since $\bD$ and $\bD^c$ are complementary, the two codes $\C_q(\bD)$ and  $\C_q(\bD^c)$ 
should be related. Below we present a few relations between the two codes. We assume that 
the columns of both incidence matrices are indexed by the points in the same order.  

\begin{theorem}\label{thm-nyear19} 
Let notation be the same as before. Let $\bone$ denote the all-one vector. 
\begin{itemize}
\item If $\bone \in \C_q(\bD)$ and $\bone \not\in \C_q(\bD^c)$, then $ \C_q(\bD) \supseteq \C_q(\bD^c)$ and $\dim(\C_q(\bD))=\dim(\C_q(\bD^c))+1$.  

\item If $\bone \in \C_q(\bD^c)$ and $\bone \not\in \C_q(\bD)$, then $ \C_q(\bD^c) \supseteq \C_q(\bD)$ and $\dim(\C_q(\bD^c))=\dim(\C_q(\bD))+1$. 

\item If $\bone \in \C_q(\bD) \cap \C_q(\bD^c)$, then $\C_q(\bD^c) = \C_q(\bD)$. 

\item If $\bone \not\in \C_q(\bD) \cup \C_q(\bD^c)$, then $ \C_q(\bD) \not\subseteq \C_q(\bD^c)$ 
and $ \C_q(\bD^c) \not\subseteq \C_q(\bD)$. In addition, 
$$ 
\C_q(\bD) \cap \C_q(\bD^c) =\left\{ \sum_{i} b_i (\bone - \bg_i): b_i \in \gf(q), \ \sum_{i} b_i =0 \right\},  
$$ 
where $\bg_i$ is the $i$-th row vector in the incidence matrix of $\bD$. 
\end{itemize} 
\end{theorem} 

\begin{proof}
By definition, $\bone - \bg_1, \cdots, \bone - \bg_b$ are the rows of the incidence 
matrix of $\bD^c$. 

Assume that $\bone \in \C_q(\bD)$ and $\bone \not\in \C_q(\bD^c)$. Then  
$\bone - \bg_1, \cdots, \bone - \bg_b$ are codewords of $ \C_q(\bD)$ It then 
follows that $ \C_q(\bD) \supseteq \C_q(\bD^c)$. Clearly, 
$\bone, \bone - \bg_1, \cdots, \bone - \bg_b$ generate   $\bg_1, \cdots, \bg_b$. 
Since $\bone \not\in \C_q(\bD^c)$, $\dim(\C_q(\bD))=\dim(\C_q(\bD^c))+1$.  

The conclusion of the second part is symmetric to that of the first part. 
The conclusion of the third part follows from the proof of the first conclusion. 

Finally, we prove the conclusions of the fourth part. On the contrary, suppose that 
$ \C_q(\bD) \subseteq \C_q(\bD^c)$. Then $\bg_1 \in  \C_q(\bD^c)$, But $\bone - \bg_1$ 
is also a codeword of $\C_q(\bD^c)$. Consequently $\bone=\bg_1+(\bone-\bg_1)$ is a 
codeword of $\C_q(\bD^c)$, which is contrary to the assumption. Consequently, 
$ \C_q(\bD) \not\subseteq \C_q(\bD^c)$. By symmetry,  $ \C_q(\bD^c) \not\subseteq \C_q(\bD)$. 

Let $\bg \in \C_q(\bD) \cap \C_q(\bD^c)$. Since $\bg \in  \C_q(\bD)$, there are $c_i \in \gf(q)$ such 
that $\bg=\sum_{i} c_i \bg_i$. Similarly,   there are $a_i \in \gf(q)$ such 
that $\bg=\sum_{i} a_i (\bone - \bg_i)$. As a result, 
$$ 
\bg=\sum_{i} c_i \bg_i=\sum_{i} a_i (\bone - \bg_i). 
$$ 
We then deduce that 
$$ 
(\sum_{i} a_i) \bone= \sum_{i} (c_i+a_i) \bg_i \in \C_q(\bD). 
$$ 
By assumption, $\bone \not\in \C_q(\bD)$. It then follows that $\sum_{i} a_i=0$. We then deduce that 
$$
\C_q(\bD) \cap \C_q(\bD^c)  \subseteq \left\{ \sum_{i} a_i (\bone - \bg_i): a_i \in \gf(q), \ \sum_{i} a_i =0 \right\}. 
$$ 
On the other hand, 
it is easily seen that 
$$
\C_q(\bD) \cap \C_q(\bD^c)  \supseteq \left\{ \sum_{i} a_i (\bone - \bg_i): a_i \in \gf(q), \ \sum_{i} a_i =0 \right\}. 
$$ 
The desired equality of the two sets finally follows. This completes the proof of this theorem. 
\end{proof}

Theorem \ref{thm-nyear19} is a refined and slightly extended result of the fact $\C_q(\bD) + \gf(q)\bone = 
\C_q(\bD^c) + \gf(q)\bone$ pointed out in \cite[p. 46]{AK92}. 
It will play a vital role in this paper. It says that in the first three cases the two 
codes  $ \C_q(\bD)$ and $\C_q(\bD^c)$ are closely related. Sometimes, it may be very hard to study 
$ \C_q(\bD)$ directly, but it may be possible to investigate $\C_q(\bD^c)$. One can then get information 
on $\C_q(\bD)$ from information on $\C_q(\bD^c)$. This is a key idea employed in this paper. To make use of this idea, 
we first need to know if $\bone \in  \C_q(\bD)$ or $\bone \in  \C_q(\bD^c)$. This could be a hard problem 
itself. For instance, it took ten years to settle this problem for the binary linear codes of a class of symmetric 
designs \cite{MW98}.   
In the last case 
(i.e., $\bone \not\in \C_q(\bD) \cup \C_q(\bD^c$), the two codes $ \C_q(\bD)$ and $\C_q(\bD^c)$ 
are loosely related.

\begin{theorem}\label{thm-june81}
Let $q$ be a power of a prime $p$. Let $\bD$ be a $t$-$(v, k, \lambda)$ design with $t \geq 2$. 
Put  
$$ 
\lambda_1=\lambda \frac{\binom{v-1}{t-1}}{\binom{k-1}{t-1}}. 
$$ 
If $\lambda_1 \not\equiv 0 \pmod{p}$, then the all-one vector $\bone$ is a codeword in 
$\C_q(\bD)$.  
\end{theorem}

\begin{proof}
It is known that $\bD$ is also a $1$-$(v, k, \lambda_1)$ design. Consequently, every point 
is incident with $\lambda_1$ blocks. It then follows that the sum over $\gf(q)$ of the row 
vectors of the incidence matrix of $\bD$ is 
$$ 
(\lambda_1, \lambda_1, ..., \lambda_1)= (\lambda_1 \bmod{p}) \bone, 
$$ 
which is a codeword in $\C_q(\bD)$. Therefore, $\bone \in \C_q(\bD)$.
\end{proof}

Theorem \ref{thm-june81} will be employed in this paper shortly, and it is quite useful. We 
inform that the condition $\lambda_1 \not\equiv 0 \pmod{p}$ is not necessary for $\bone$ 
being a codeword of $\C_q(\bD)$.

\subsection{Relations between $\C_p(\bD)$ and $\C_q(\bD)$} 

Let $q=p^s$, where $s \geq 2$ and $p$ is a prime. Let $\bD$ be a $t$-$(v, k, \lambda)$ design. 
In this section, we document some relations between $\C_p(\bD)$ and $\C_q(\bD)$. 

\begin{theorem}\label{thm-feb161}
Let $q=p^s$, where $s \geq 2$. Let $\bD$ be a $t$-$(v, k, \lambda)$ design. Then 
$\C_p(\bD)$ is the subfield subcode over $\gf(p)$ of $\C_q(\bD)$. Further, 
$$ 
\C_p(\bD)^\perp = \tr(\C_q(\bD)^\perp), 
$$
where $\tr(\C_q(\bD)^\perp)$ denotes the trace code of $\C_q(\bD)^\perp$.  
\end{theorem} 

\begin{proof}
Let $\bm_1, \bm_2, ..., \bm_b$ be the row vectors in the incidence matrix of $\bD$. 
Let $\alpha$ be a generator of $\gf(q)^*$. Let $u_i = \sum_{j=0}^{s-1} u_{ij} \alpha^j  
\in \gf(q)$ for all $1 \leq i \leq b$, where all $u_{ij} \in \gf(p)$. We have then 
\begin{eqnarray*}
\sum_{i=1}^b u_i \bm_i=\sum_{i=1}^b \left( \sum_{j=0}^{s-1} u_{ij} \alpha^j  \right) \bm_i 
= \sum_{j=0}^{s-1} \left( \sum_{i=1}^b u_{ij} \bm_i  \right) \alpha^j.  
\end{eqnarray*}   
Since each $\bm_i \in \gf(p)^v$, $\sum_{i=1}^b u_{ij} \bm_i$ is a vector in $\gf(p)^v$ 
for each $j$. It then follows that 
$\sum_{i=1}^b u_i \bm_i \in \gf(p)^v$ if and only if 
\begin{eqnarray}\label{eqn-feb161}
\sum_{i=1}^b u_{ij} \bm_i = \bzero, \mbox{ for all $1 \leq j \leq s-1$.}  
\end{eqnarray} 
If the system of equations in (\ref{eqn-feb161}) holds, then 
$$ 
\sum_{i=1}^b u_i \bm_i=\sum_{i=1}^b u_{i0} \bm_i.  
$$  
Consequently, $\C_p(\bD)$ is the subfield subcode over $\gf(p)$ of $\C_q(\bD)$. 
The last desired result then follows from Delsarte's theorem \cite{Delsarte75}. 
\end{proof}

\begin{theorem}\label{thm-feb162}
Let $q=p^s$, where $s \geq 2$. Let $\bD$ be a $t$-$(v, k, \lambda)$ design. Then  
$ 
\C_p(\bD) = \tr(\C_q(\bD)).  
$ 
\end{theorem} 

\begin{proof}
Let $\bm_1, \bm_2, ..., \bm_b$ be the row vectors in the incidence matrix of $\bD$. 
Note that each $\bm_i \in \gf(q)^v$ and each codeword of $\C_q(\bD)$ can be expressed 
as $\sum_{i=1}^b u_i \bm_i$, where $u_i \in \gf(q)$. We have 
$$ 
\tr\left(\sum_{i=1}^b u_i \bm_i\right)= \sum_{i=1}^b \tr(u_i) \bm_i. 
$$ 
When $u_i$ ranges over the elements in $\gf(q)$, $\tr(u_i)$ ranges over each element 
of $\gf(p)$ eaxctly $p^{s-1}$ times. The desired conclusion then follows.    
\end{proof}

\begin{theorem}\label{thm-Cp-Cq}
Let $q=p^s$, where $s \geq 2$. Let $\bD$ be a $t$-$(v, k, \lambda)$ design. 
Then  $\mathrm{dim}_{\mathrm{GF}(p)}(\C_p(\bD))=\mathrm{dim}_{\mathrm{GF}(q)}(\C_q(\bD))$ and 
$d(\C_p(\bD))=d(\C_q(\bD))$, where $d(\C)$ denotes the minimum distance of the code $\C$. 
\end{theorem}

\begin{proof}
Let $M_\bD=(m_{ij})$ be  the incidence matrix  of $\bD$. Since $\mathrm{dim}_{\mathrm{GF}(p)}(\C_p(\bD))$, $\mathrm{dim}_{\mathrm{GF}(q)}(\C_q(\bD))$
 both equal to the rank of the matrix $M_\bD$, then $\mathrm{dim}_{\mathrm{GF}(p)}(\C_p(\bD))=\mathrm{dim}_{\mathrm{GF}(q)}(\C_q(\bD))$. 
 
 By Theorem \ref{thm-feb161}, $d(\C_p(\bD))\ge d(\C_q(\bD))$. Let $\{\mathbf e_1, \cdots, \mathbf e_{s}\}$ be a basis of $\gf(q)$ over $\gf(p)$. 
 Let $\mathbf{c}=(c_1, \cdots, c_{v})$ be any nonzero codeword in 
 $\C_q(\bD)$. Then, there are $\alpha_1, \cdots, \alpha_b \in \gf(q)$ such that $c_j=\sum_{i=1}^{b} m_{ij} \alpha_i$. 
 Let $\alpha_i= \sum_{t=1}^s a_{it} \mathbf e_t$, where $a_{it}\in \gf(p)$. Then 
 $c_j=\sum_{i=1}^{b}  \sum_{t=1}^s m_{ij}  a_{it} \mathbf e_t$ and $\mathbf c=\sum_{t=1}^s \mathbf e_t \mathbf c^t$
 , where $\mathbf c^t=(\sum_{i=1}^{b}  \sum_{t=1}^s m_{i1}  a_{it} , \cdots, \sum_{i=1}^{b}  \sum_{t=1}^s m_{iv}  a_{it} ) \in \C_p(\bD)$.
 There is a $t_0$ such that $\mathbf c^{t_0}\neq \mathbf 0$, as $\mathbf c \neq \mathbf 0$.
 Then $\mathrm{wt}(\mathbf c)\ge \mathrm{wt}(\mathbf c^{t_0})\ge d(\C_p(\bD))$. Thus $d(\C_q(\bD)) \ge d(\C_p(\bD))$.
 This completes the proof. 
\end{proof}

Theorem \ref{thm-Cp-Cq} explains why we restrict ourself to $\C_p(\bD)$ rather than treating $\C_q(\bD)$ in this paper, though the two codes have different weight distributions.

\subsection{A general result about $\C_q(\bD)^\perp$} 

The next result is useful \cite[p. 54]{AK92}, and will be used later in this paper. 

\begin{theorem}\label{thm-AKp54}
Let $\bD = (\cP, \cB)$ be a $2$-$(v,k,\lambda)$ design with $k<v$. 
If $\C_{q}(\bD) \neq \gf(q)^v$, then the minimum weight of $\C_{q}(\bD)^\perp$ 
is at least 
$$ 
\frac{v-1}{k-1}+1. 
$$ 
\end{theorem} 

The lower bound in Theorem \ref{thm-AKp54} is not tight in general, but reasonably good in some special 
cases.

\section{The binary case}\label{sec-binarycase}

In this section, we present some fundamental results about binary codes and their designs, which do not hold in general for nonbinary codes. 

\begin{theorem}\label{thm-jan21}
Let $\C$ be an $[n, k, d]$ binary code which holds designs. Let $\bD_i(\C)$ denote the 
support design of the codewords of weight $i$ in $\C$, where the point set is the set 
of coordinates, i.e., $\{0,1, \ldots, n-1\}$. Let $\C_2(\bD_i(\C))$ denote the 
 binary code of the design $\bD_i(\C)$, where the point set is the ordered 
set $\{0,1, \ldots, n-1\}$. Then the following statements are true. 
\begin{enumerate}
\item $\C_2(\bD_i(\C))$ is a subcode of $\C$ and they are equal if and only if the codewords of 
      weight $i$ span $\C$.  
\item $\Aut(\C) \leqslant \Aut(\bD_i(\C))$, i.e., the former is a subgroup of the latter.  
\item If the codewords of weight $i$ or the codewords of weight $i$ and the all-one 
      vector $\bone$ generate $\C$, then $\Aut(\C) = \Aut(\bD_i(\C))$.      
\end{enumerate}  
\end{theorem} 

\begin{proof} 
Notice that $\C$ is a binary linear code. 
By definition, each row of the incidence matrix of the design $\bD_2(\C)$ is a codeword 
of $\C$. Consequently, $\C_2(\bD_i(\C))$ is a subcode of $\C$. The desired first conclusion 
then follows.  

Recall that the automorphism group $\Aut(\C)$ of a binary linear code $\C$ is its permutation 
automorphism group $\PAut(\C)$. Any $\sigma \in \Aut(\C)$ is clearly a permutation of the 
coordinates of the codewords in $\C$ that fixes $\C$, and is thus a permutation of the point set and block 
set of $\bD_i(\C)$. This proves the conclusion of the second part.  

We now prove the conclusion of the third part. 
Note that any permutation of $\{0,1, \ldots, n-1\}$ fixes the all-one vector $\bone$. 
By assumption, every codeword of $\C$ is a linear combination of the rows of the 
incidence matrix of $\bD_i(\C)$ and the all-one vector. Then by assumption, any $\sigma 
\in \Aut(\bD_i(\C))$ is an element of $\Aut(\C)$. Therefore, 
$\Aut(\bD_i(\C)) \leqslant \Aut(\C)$. The desired third conclusion then follows from 
that of the second part. 
\end{proof} 

The proof of Theorem \ref{thm-jan21} showed that $\C_2(\bD_i(\C))$ is a subcode 
of the original code $\C$. Regarding these two codes, we have the following comments: 
\begin{enumerate}
\item $\C$ and $\C_2(\bD_i(\C))$ have the same length, but may have different dimensions 
      and minimum distances. In general, 
      $$ 
      \dim(\C_2(\bD_i(\C))) \leq \dim(\C) \mbox{ and } d(\C_2(\bD_i(\C))) \geq d(\C),  
      $$ 
      where $d(\C)$ denotes the minimum distance of $\C$. 
\item Let $\C$ be an $[n, k, d]$ binary code. Then $\C_2(\bD_d(\C))$ has parameters 
      $[n, k', d]$ with $k' \leq k$. When $k' < k$, $\C_2(\bD_d(\C))$ is not as good as 
      the original code $\C$. However, the dual code $\C_2(\bD_i(\C))^\perp$ may be better 
      then $\C^\perp$, as it may happen that  
      $$ 
      \dim(\C_2(\bD_d(\C))^\perp) > \dim(\C^\perp) \mbox{ and }  
      d(\C_2(\bD_d(\C))^\perp) = d(\C^\perp).  
      $$          
\item  Let $\C$ be an $[n, k, d]$ binary code. Let $i$ be an integer such that $d < i < n$ 
and $\bD_i(\C)$ is a $2$-design. Then $\C_2(\bD_i(\C))$ has parameters  $[n, k', d']$ 
with $k' \leq k$ and $d' \geq d$. The code $\C_2(\bD_i(\C))$ could be optimal and thus 
interesting. The following Example \ref{exam-june6} justifies this claim. 
\end{enumerate} 
Hence, the code $\C_2(\bD_i(\C))$ or its dual could be interesting in many cases. 

\begin{example} \label{exam-june6}
Let $m$ be a positive integer. For each $(a, b, h) \in \gf(2^{m}) \times \gf(2^{2m}) \times \gf(2)$, 
define a Boolean function from $\gf(2^{2m})$ to $\gf(2)$ by  
$$ 
f_{(a,b,h)}(x)=\tr_{m/1}\left[  a \tr_{2m/m}\left(u x^{1+2^{m-1}} \right)  \right] + \tr_{2m/1}(bx) + h, 
$$ 
where $\tr_{j/i}$ is the trace function from $\gf(2^j)$ to $\gf(2^i)$ and $u \in \gf(2^{2m}) \setminus \gf(2^{m})$.  
Define a linear code 
\begin{eqnarray*}
\C(m)=\left\{  (f_{(a,b,h)}(x))_{x \in \gf(2^{2m})}: a \in \gf(2^m), \ b \in \gf(2^{2m}), \ h \in \gf(2) \right\}.  
\end{eqnarray*}
It is shown in \cite{DMT19} that $\C_m$ has parameters $[2^{2m}, 3m+1, 2^{2m-1}-2^{m-1}]$ 
and weight enumerator 
\begin{eqnarray*}\label{eqn-wtenumerator111}
1 + (2^m-1)2^{2m} z^{2^{2m-1} - 2^{m-1}} + 2(2^{2m}-1)z^{2^{2m-1}} 
+ (2^m-1)2^{2m} z^{2^{2m-1} + 2^{m-1}} + z^{2^{2m}}.  
\end{eqnarray*}  
In addition, we have the following \cite{DMT19}: 
\begin{enumerate}
\item The codewords of minimum weight of $\C(m)$ hold a 2-design $\bD_{2^{2m-1} - 2^{m-1}}(\C(m))$
with parameters 
\begin{equation*}
\label{par}
2-(2^{2m}, 2^{2m-1} - 2^{m-1}, (2^m -1)(2^{2m-2} - 2^{m-1})).
\end{equation*}
\item The codewords of  weight $2^{2m-1} + 2^{m-1}$  of $\C(m)$ hold a 
2-design $\bD_{2^{2m-1} + 2^{m-1}}(\C(m))$ with parameters
\begin{equation*}\label{prc}
 2-(2^{2m}, 2^{2m-1} + 2^{m-1}, (2^m - 1)(2^{2m-2} + 2^{m-1})).
 \end{equation*}
 \item The codewords of  weight $2^{2m-1}$ of $\C(m)$ hold a 2-design $\bD_{2^{2m-1}}(\C(m))$ with parameters 
 \begin{equation*}\label{prc2}
 2-(2^{2m}, 2^{2m-1} , 2^{2m-1}-1), 
 \end{equation*} 
 which is actually a $3$-design. 
 \end{enumerate}
It is proved in \cite{DMT19} that the minimum weight codewords generate $\C(m)$. It then follows from 
Theorem \ref{thm-jan21} that $\C_2(\bD_{2^{2m-1} - 2^{m-1}}(\C(m)))=\C(m)$. It is easily seen that 
$\C_2(\bD_{2^{2m-1}}(\C(m)))$ is the first-order Reed-Muller code, which is optimal. This demonstrates 
that studying the binary code $\C_2(D_i(\C))$ for some $i$ could be interesting. 
\end{example}

\begin{theorem}\label{thm-jan22}
Let $\bD$ be a design. Then $\Aut(\bD) \leqslant \Aut(\C_2(\bD))$. 
\end{theorem} 

The proof of this theorem is straightforward. The equality in Theorem \ref{thm-jan22} 
may be valid in some special cases. 
The following theorem follows from Theorem \ref{thm-jan22} and the second part of 
Theorem \ref{thm-jan21}. 

\begin{theorem}
Let $\C$ be an $[n, k, d]$ binary code which holds designs. Let $\bD_i(\C)$ denote the 
support design of the codewords of weight $i$ in $\C$, where the point set is the set 
of coordinates, i.e., $\{0,1, \ldots, n-1\}$. Let $\C_2(\bD_i(\C))$ denote the 
 binary code of the design $\bD_i(\C)$, where the point set is the ordered 
set $\{0,1, \ldots, n-1\}$. Then 
$$ 
\Aut(\C) \leqslant \Aut(\C_2(\bD_i(\C))). 
$$
\end{theorem}

\section{The code of the design held in the Simplex code}\label{sec-PRMcode}

Our task in this section is to study the code of the design held in the Simplex code. 
To this end, we have to introduce some known results about the codes of the designs 
in the projective geometry $\PG(m-1, q)$ and the projective Reed-Muller codes in Section \ref{sec-june261a}, 
as they are needed in Section \ref{sec-june261b}.  Hence, Section \ref{sec-june261a} below is not meant to be a survey, 
but a recall of some auxiliary results needed in Section \ref{sec-june261b}.

\subsection{The codes of the designs in the projective geometry $\PG(m-1, q)$}\label{sec-june261a}

The points of the \emph{projective space\index{projective space}} (also called \emph{projective geometry\index{
projective geometry}}) $\PG(m-1, q)$ are all the 1-dimensional subspaces of the vector space $\gf(q)^{m}$;  
the lines are the 2-dimensional  subspaces of $\gf(q)^{m}$, the planes are the 3-dimensional subspaces of 
$\gf(q)^{m}$, and the hyperplanes  are the $(m-1)$-dimensional  subspaces of $\gf(q)^{m}$; and incidence is the set-theoretic inclusion. The elements of the projective space 
$\PG(m-1,q)$ are the points, lines, planes, ..., and the hyperplanes. But  
the space $\gf(q)^{m}$ is not an element of $\PG(m-1,q)$, as it contains every other subspace and thus plays no role. The \emph{projective dimension} of an element in $\PG(m-1,q)$  
is one less than that of the corresponding element in the vector space $\gf(q)^{m}$. 
The $d$-flats in the projective geometry $\PG(m-1, q)$ form a $2$-design, which is documented 
below and is well known in the literature \cite{BJL}. 

\begin{theorem}\label{thm-PGdesigns}
Let $\cB$ denote the set of all $d$-flats in 
$\PG(m-1, q)$, and $\cP$ the point set of $\PG(m-1, q)$, and the incidence relation $\cI$ is the containment relation. Then the triple $\PG_d(m-1, q):=(\cP, \cB, \cI)$ is a $2$-$(v, k, \lambda)$ design, where   
$$ 
v=\frac{q^m-1}{q-1}, \ k = \frac{q^{d+1}-1}{q-1}, \ \lambda= \left[ \myatop{m-2}{d-1} \right]_q.  
$$ 
In addition, the number of blocks in this design is 
$$ 
b=\left[ \myatop{m}{d+1} \right]_q. 
$$
In particular, $\PG_1(m-1, q)$ is a Steiner system\index{Steiner system} $S(2, q+1, (q^m-1)/(q-1))$, and $\PG_{m-2}(m-1, q)$ is a symmetric design with parameters 
$$ 
2-\left(\frac{q^m-1}{q-1}, \ \frac{q^{m-1}-1}{q-1}, \ \frac{q^{m-2}-1}{q-1}  \right)  
$$ 
for $m \geq 3$. 
\end{theorem}

Let $q$ be a prime power and let $m \geq 2$. 
A point of the projective geometry $\PG(m-1, \gf(q))$ is given in homogeneous coordinates by $(x_0, x_1,\ldots, x_{m-1})$ where all $x_i$ are in $\mathrm{GF}(q)$ and are not all zero;
each point has $q-1$ coordinate representations, since $(ax_0,ax_1,...,ax_{m-1})$ and
$(x_0,x_1,...,x_{m-1})$
yield the same $1$-dimensional subspace of $\mathrm{GF}(q)^{m}$ for any non-zero $a\in \mathrm{GF}(q)$.

For an integer $r \ge 0$, let $\mathrm{PP}(r,m-1,q)$ denote the linear subspace of $\mathrm{GF}(q)[x_0,x_1, \dots, x_{m-1}]$ that is spanned by all
monomial $x_0^{i_0}x_1^{i_1}\cdots x_{m-1}^{i_{m-1}}$ satisfying the following two conditions: 
\begin{itemize}
\item $\sum_{j=0}^{m-1} i_j \equiv 0 \pmod{q-1}$,
\item    $0<\sum_{j=0}^{m-1} i_j   \le r(q-1)$.
\end{itemize} 
Each $a \in \mathrm{GF}(q)$ is viewed as the constant function $f_a(x_0, x_1, \ldots, x_{m-1}) \equiv a$.

Let $\{\mathbf{x}^1, \dots, \mathbf{x}^{N}\}$
be the set of projective points in $\mathrm{PG}(m-1,q)$, where $N=\frac{q^m-1}{q-1}$. Then, the \emph{$r$th order
 projective generalized Reed-Muller code}\index{projective generalised Reed-Muller code} $\mathrm{PRM}(r,m-1,q)$ of length $\frac{q^m-1}{q-1}$ is defined by 
\begin{align*}
%\label{eq:PRM}
\mathrm{PRM}(r,m-1,q)=\left \{\left (f(\mathbf{x}^1), \dots, f(\mathbf{x}^N) \right ): f\in \mathrm{PP}(r,m-1,q)\cup \mathrm{GF}(q) \right \}.
\end{align*}
When $r\geq 1$, let  $\mathrm{PRM}^*(r,m-1,q)$ be the subcode of $\mathrm{PRM}(r,m-1,q)$
defined by
\begin{align*}
%\label{eq:PRM*}
\mathrm{PRM}^*(r,m-1,q)=\left \{\left (f(\mathbf{x}^1), \dots, f(\mathbf{x}^N) \right ): f\in \mathrm{PP}(r,m-1,q) \right \}.
\end{align*}
Thus, $\mathrm{PRM}^*(r,m-1,q)$ is a subcode of $\mathrm{PRM}(r,m-1,q)$.
For the minimum weight and the dual of the projective generalized Reed-Muller code, we have the following \cite{AK98}.

\begin{theorem}\label{thm-feb141}
Let $0 \le r \le m-1$. Then, the minimal weight of $\mathrm{PRM}(r,m-1,q)$ is $\frac{q^{m-r}-1}{q-1}$ and
\begin{align*}
\mathrm{PRM}(r,m-1,q)^{\perp} = \mathrm{PRM}^*(m-1-r,m-1,q).
\end{align*}
\end{theorem}

Let $p$ be a prime. Then the relation between the codes $\mathsf{C}_{p}(\PG_{r-1}(m-1,p))$ of the designs 
of projective geometries over $\mathrm{GF}(p)$ and the projective generalized Reed-Muller codes over $\mathrm{GF}(p)$ is 
given as follows \cite{AK98}. 

\begin{theorem}\label{thm:PG-PRM}
Let $m$ be a positive integer, $p$ a prime, and $1 \le r \le m$.

\rm(i) The code $\mathsf{C}_{p}(\PG_{r-1}(m-1,p))$ 
from the design of points and projective
$(r-1)$-dimensional subspaces of the projective geometry $\PG(m-1,p)$ is the same as $\mathrm{PRM}(m-r,m-1,p)$
 up to a permutation of coordinates.

 \rm(ii) $\mathsf{C}_{p}(\PG_{r-1}(m-1,p))$ has minimum weight
 $\frac{p^r-1}{p-1}$ and the minimum-weight vectors are the multiples of the characteristic  vectors of the blocks.

 \rm(iii) The dual code $\mathsf{C}_{p}(\PG_{r-1}(m-1,p))^\perp$
  is the same as $\mathrm{PRM}^*(r-1,m-1,p)$
 up to a permutation of coordinates and has minimum weight at least  $\frac{p^{m-r+1}-1}{p-1}+1$.

\rm(iv) The dimension of the code $\mathsf{C}_{p}(\PG_{r-1}(m-1,p))$ is
$$ \frac{p^m-1}{p-1}-\sum_{i=0}^{r-2} (-1)^{i} \binom{(r-1-i)(p-1)-1}{i} \binom{m-r+(r-1-i)p}{m-1-i}.$$
\end{theorem}

To obtain the codes of the designs coming from projective spaces over $\mathrm{GF}(q)$ with $q=p^s$, we need to restrict the codes $\mathrm{PRM}(m-r,m-1,q)$
 to subfield subcodes. Let $\mathsf C$ be a linear code over $\mathrm{GF}(q)$. The set $\mathsf C_{q/p}$ of vectors in $\mathsf C$, all of whose coordinates lie in $\mathrm{GF}(p)$, is called the subfield subcode of $\mathsf C$ over $\mathrm{GF}(p)$. Denote by $\mathrm{PRM}_{q/p}(m-r,m-1,q)$ the subfield subcode of the projective generalized Reed-Muller code $\mathrm{PRM}(m-r,m-1,q)$.  Then the relation between the codes $\mathsf{C}_{p}(\PG_{r-1}(m-1,q))$ of the designs 
of projective geometries over $\mathrm{GF}(q)$ and the subfield subcode $\mathrm{PRM}_{q/p}(m-r,m-1,q)$ of the projective generalized Reed-Muller code is 
given as follows \cite{AK98}. 

 \begin{theorem}\label{thm:PG-PRM-q}
 Let m be any positive integer, $q=p^s$ where $p$ is a prime, and let $1 \le r \le m$.
 
 \rm(i) The code $\mathsf{C}_{p}(\PG_{r-1}(m-1,q))$ 
from the design of points and projective
$(r-1)$-dimensional subspaces of the projective geometry $\PG(m-1,q)$ is the same as $\mathrm{PRM}_{q/p}(m-r,m-1,q)$
 up to a permutation of coordinates.

 \rm(ii) $\mathsf{C}_{p}(\PG_{r-1}(m-1,q))$ has minimum weight
 $\frac{q^r-1}{q-1}$ and the minimum-weight vectors are the multiples of the characteristic  vectors of the blocks.

 \rm(iii) The dual code $\mathsf{C}_{p}(\PG_{r-1}(m-1,q))^\perp$
   has minimum weight at least  $\frac{q^{m-r+1}-1}{q-1}+1$.

\rm(iv) The dimension of the code $\mathsf{C}_{p}(\PG_{m-2}(m-1,q))$ is
$$   \binom{p+m-2}{m-1}^s +1.$$

 \end{theorem}

Serre has proved in \cite{Ser89} the following inequality, conjectured by Tsfasman:
\begin{theorem}
Let $m\ge 2$ and $f$ be a nonzero homogeneous polynomial in $\mathrm{GF}(q)[x_0, x_1, \cdots, x_{m-1}]$ with $\mathrm{deg}(f)\le q+1$.
Let $N_f= | \{ \mathbf x \in \mathrm{PG}(m-1, q): f(\mathbf x)=0\}|$.
Then
$$ N_f\le \mathrm{deg}(f)q^{m-2}+\frac{q^{m-2}-1}{q-1}.$$
Moreover, if $\mathrm{deg}(f) \le q$, the upper bound  is attained only if the set
$\{ \mathbf x \in \mathrm{PG}(m-1, q): f(\mathbf x)=0\}$ is a
union of $\mathrm{deg}(f)$ hyperplanes whose intersection contains a subspace of codimension $2$.
\end{theorem}

Taking $\mathrm{deg}(f)=q-1$, we have the following result, 
\begin{theorem}\label{thm:2q-power}
Let $m\ge 2$. Then $\mathrm{PRM}^*(1,m-1,q)$ has minimum weight
 $2 q^{m-2}$.
\end{theorem}

Combining Theorems \ref{thm-feb141} and \ref{thm:PG-PRM}, we have 
$$ 
\mathrm{PRM}^*(1,m-1,p)=\mathsf{C}_{p}(\PG_{1}(m-1,p))^\perp=\mathrm{PRM}(m-2,m-1,p)^\perp,  
$$ 
where the equalities mean the equivalence of codes. By definition, $\mathrm{PRM}^*(1,m-1,p)$ 
is a subcode of $\mathrm{PRM}(1,m-1,p)$. 

\subsection{The code of the design held in the Simplex code}\label{sec-june261b}

We view $\gf(q^m)$ as an $m$-dimensional vector space over $\gf(q)$. Let $\alpha$ be a 
generator of $\gf(q^m)^*$. Then 
$$ 
\cP=\{1, \alpha, \alpha^2, ..., \alpha^{v-1}\}=\gf(q^m)^*/\gf(q)^* 
$$ 
is the set of points in the projective geometry $\PG(m-1, q)$, where $v=(q^m-1)/(q-1)$. 

By the definition $\alpha$ and $v$, it is easily seen that 
\begin{eqnarray}\label{eqn-feb141}
\left\{(\tr(a \alpha^i))_{i=0}^{v-1}: a \in \gf(q^m)\right\}  
\end{eqnarray} 
is the Simplex code whose dual is the Hamming code. 
Clearly, the weight enumerator of the Simplex code is given by 
\begin{eqnarray}\label{eqn-feb142}
1+(q^m-1)z^{q^{m-1}}. 
\end{eqnarray}  
By the Assmus-Mattson theorem,  
the codewords of weight $q^{m-1}$ in the Simplex code form a design 
$\bD$ with the following parameters 
\begin{eqnarray}
2\mbox{ --} \left( \frac{q^m-1}{q-1}, \ q^{m-1}, \ (q-1)q^{m-2}  \right).  
\end{eqnarray} 
Our objective in this section is to study the code $\C_q(\bD)$. Note that the design $\bD$ 
is not a geometric design in the projective geometry $\PG(m-1, q)$. Hence, we are not able 
to apply Theorem \ref{thm:PG-PRM} directly, but we will make use of it indirectly. To this end, 
we need to do some preparations. 

\begin{lemma} 
The complementary design $\bD^c$ of $\bD$ is the geometric design $\PG_{m-2}(m-1,q)$ with 
parameters 
\begin{eqnarray}
2\mbox{ --} \left( \frac{q^m-1}{q-1}, \ \frac{q^{m-1}-1}{q-1}, \ \frac{q^{m-2}-1}{q-1}\right).  
\end{eqnarray} 
\end{lemma}  

\begin{proof}
We use the trace expression of the Simplex code given in (\ref{eqn-feb141}), 
and index the coordinates of the code with the elements in $\gf(q^m)$. Let 
$$ 
\bc_a=(\tr(a \alpha^i))_{i=0}^{v-1}
$$ 
where $a \neq 0$. Then the complement $\support(\bc_a)^c$ of the support $\support(\bc_a)$ 
of the codeword $\bc_a$ is given by  
$$ 
\support(\bc_a)^c=\{\alpha^i: 0 \leq i \leq v-1 \mbox{ and } \tr(a \alpha^i)=0\},  
$$ 
which is a hyperplane in $\PG(m-1, q)$. On the other hand, every hyperplane in $\PG(m-1, q)$ 
is of this form and corresponds to such codeword in $\mathrm{PRM}^*(1,m-1,q)$. The desired 
conclusion then follows. 
\end{proof}

The following lemma will play an important role in proving the main result of this section. 

\begin{lemma}\label{lem-feb141} 
The code $ \C_p(\bD)^\perp$ contains the all-one vector $\bone$ and the code $\C_p(\bD)$  does not contain the all-one vector $\bone$. 
\end{lemma} 

\begin{proof}
Since each row in the incidence matrix of the design $\bD$ has Hamming weight $q^{m-1}$, 
the all-one vector $\bone$ of length $v=(q^m-1)/(q-1)$ is orthogonal to all rows in 
the incidence matrix. As a result, $\bone \in \C_p(\bD)^\perp$. Note that the inner 
product of $\bone$ and itself is $v \pmod{q}=1$. It then follows that $\bone \not\in 
\C_p(\bD)$.  
\end{proof} 

The main result of this section is the following. 

\begin{theorem}\label{thm-majorthm2}
The code $\C_p(\bD)$ of the design $\bD$ has parameters 
\begin{eqnarray*}
\left[\frac{q^m-1}{q-1}, \ \binom{p+m-2}{m-1}^s, \ d   \right], 
\end{eqnarray*} 
where 
\begin{eqnarray}\label{eqn-afeb141}
d \geq 2 q^{m-2}. 
\end{eqnarray} 
Moreover, if $q=p$, $d=2 q^{m-2}$.
\end{theorem} 

\begin{proof}
We first prove that the all-one vector $\bone$ is a codeword of $\C_p(\PG_{m-2}(m-1, q))$. 
Note that the number of blocks of the design $\PG_{m-2}(m-1,q)$ containing a point of the 
design is 
$$ 
\lambda_1=\frac{q^{m-2}-1}{q-1} \frac{\binom{ \frac{q(q^{m-1}-1)}{q-1} }{1} }{\binom{\frac{q(q^{m-1}-1)}{q-1}}{1}}=\frac{q^{m-1}-1}{q-1}. 
$$ 
Hence, $\lambda_1 \pmod{q} =1$. Consequently, the sum of the row vectors over $\gf(q)$ 
of the incidence matrix of the design $\PG_{m-2}(m-1,q)$ is 
$$ 
(\lambda_1, ..., \lambda_1)=(1,..., 1)=\bone \in \C_p(\PG_{m-2}(m-1, q)). 
$$ 
We then deduce from Theorem \ref{thm-nyear19} and Lemma \ref{lem-feb141} that 
$$ 
\C_p(\bD) \subset \C_p(\PG_{m-2}(m-1, q)) \mbox{ and } 
\dim(\C_p(\bD)) = \dim(\C_p(\PG_{m-2}(m-1, q)))-1.  
$$  
Note that 
$$ \C_p(\bD) \subseteq \mathrm{PRM}^*_{q/p}(1,m-1,q) \subset  \mathrm{PRM}_{q/p}(1,m-1,q)=  \C_p(\PG_{m-2}(m-1, q)),$$
where $\mathrm{PRM}^*_{q/p}(1,m-1,q)$ is the subfield subcode of the code $\mathrm{PRM}^*(1,m-1,q)$.
Noticing that $\dim(\C_p(\bD)) - \dim(\C_p(\PG_{m-2}(m-1, q)))=1$, one has $\C_p(\bD) = \mathrm{PRM}^*_{q/p}(1,m-1,q)$.
Then, the desired conclusions  follow from Theorems \ref{thm:PG-PRM-q} and \ref{thm:2q-power}. 
\end{proof}

Note that the lower bound on the minimum distance $d$ given in (\ref{eqn-afeb141}) 
is the minimum distance of the code $\C_p(\PG_{m-2}(m-1, q))$. Although the difference 
of the dimensions of $\C_p(\PG_{m-2}(m-1, q))$ and $\C_p(\bD)$ is only one, the 
difference between their minimum distances could be very large for $q \geq 3$. Table 
\ref{table-feb141} 
documents the parameters of the two codes in some cases. When $m=2$, both code are MDS and optimal. 
When $(q, m)=(3,3)$, the code $\C_p(\bD)$ has parameters $[13,6,6]$ and is optimal. 
Note that the code $\C_p(\bD)$ is much better than $\C_p(\PG_{m-2}(m-1, q))$ in many 
cases in terms of error correcting capability.

\begin{table}[htbp]
\centering
\caption{The parameters of $\C_p(\bD)$ and $\C_p(\PG_{m-2}(m-1, q))$}
\label{table-feb141}
\begin{tabular}{|c|c|c|}
  \hline
$(q,m)$ & $\C_p(\bD)$  & $\C_p(\PG_{m-2}(m-1, q))$ \\  \hline
$(3,2)$ & $[4,3,2]$    & $[4,4,1]$  \\ \hline 
$(3,3)$ & $[13,6,6]$    & $[13,7,4]$  \\ \hline 
$(3,4)$ & $[40,10,18]$    & $[40,11,13]$  \\ \hline 
$(3,5)$ & $[121,15,54]$    & $[121,16,40]$  \\ \hline 
$(4,2)$ & $[5,4,2]$    & $[5,5,1]$  \\ \hline 
$(4,3)$ & $[21,9,8]$    & $[21,10,5]$  \\ \hline 
$(4,4)$ & $[85,16,32]$    & $[85,17,21]$  \\ \hline 
$(5,2)$ & $[6,5,2]$    & $[6,6,1]$  \\ \hline 
$(5,3)$ & $[31,15,10]$    & $[31,16,6]$  \\ \hline 
\end{tabular}
\end{table}
 
 In fact, experimental data strongly supports the following conjecture. 

\begin{conj}\label{conj-june261} 
Let $\bD$ be defined as before.  
The minimum distance of the code $\C_p(\bD)$ equals $2q^{m-2}$.  
\end{conj}

\begin{theorem}\label{thm-majorthm2dual} 
Let $\bD$ be defined as before.  
The dual code $\C_p(\bD)^\perp$ has parameters 
\begin{eqnarray*}
\left[\frac{q^m-1}{q-1}, \ \frac{q^m-1}{q-1}-\binom{p+m-2}{m-1}^s, \ d^\perp     \right], 
\end{eqnarray*} 
where $d^\perp \geq 3$. Moreover, if $q=p$, $d^\perp= p+1$.
\end{theorem} 

\begin{proof}
The dimension of the code $\C_p(\bD)^\perp$ follows from Theorem \ref{thm-majorthm2dual}. 
Note that the design $\bD$ has parameters 
\begin{eqnarray}
2\mbox{ --} \left( \frac{q^m-1}{q-1}, \ q^{m-1}, \ (q-1)q^{m-2}  \right).  
\end{eqnarray} 
The desired lower bound on $d^\perp$ follows from Theorem \ref{thm-AKp54}. 

When $q=p$, from the proof of Theorem \ref{thm-majorthm2}, $\C_p(\bD)= \mathrm{PRM}^*(1,m-1,p)$.
By Theorem \ref{thm-feb141}, one has $\C_p(\bD)^{\perp}=\mathrm{PRM}(m-2,q)$ and the minimum distance of the code $\C_p(\bD)^{\perp}$ equals $p+1$ by Theorem \ref{thm:PG-PRM}.
\end{proof}

In fact, experimental data strongly supports the following conjecture. 

\begin{conj}\label{conj-june262}  
Let $\bD$ be defined as before.  
The minimum distance of the code $\C_p(\bD)^\perp$ equals $q+1$.  
\end{conj}

\section{Linear codes from the $t$-designs held in the generalised Reed-Muller codes}\label{sec-gRMcode}

Our task in this section is to study the linear codes from the $t$-designs held in the generalised Reed-Muller codes. 
To this end, we have to introduce 
some known results about the codes of the designs in the affine geometry $\AG(m, q)$ and the generalised 
Reed-Muller codes in Section \ref{sec-june262a}, 
as they are needed in Section \ref{sec-june262b}.  Hence, Section \ref{sec-june262a} below is not meant to be a survey, 
but a recall of some auxiliary results needed in Section \ref{sec-june262b}.   

\subsection{The codes of the designs in the affine geometry $\AG(m, q)$}\label{sec-june262a}

The affine geometry $\AG(m, q)$, where the points are the vectors in the vector 
space $\gf(q)^m$, the lines are the cosets of all the one-dimensional subspaces, 
the planes are the cosets of the two-dimensional subspaces, the $i$-flats are 
the cosets of the $i$-dimensional subspaces, and the hyperplanes are the cosets 
of the $(m-1)$-dimensional subspaces of $\gf(q)^m$. The $d$-flats of 
$\gf(q)^m$ can be employed to construct $2$-designs. 

\begin{theorem}\cite{BJL}\label{thm-AGdesigns}
Let $\cB$ denote the sets of all $d$-flats in 
$\gf(q)^m$, and $\cP$ the set of all vectors in $\gf(q)^m$, and $\cI$ the containment relation. Then  the triple $\AG_d(m, q):=(\cP, \cB, \cI)$ is $2$-$(v, k, \lambda)$ design, where   
$$ 
v=q^m, \ k = q^d, \ \lambda=\left[  \myatop{m-1}{d-1} \right]_q, 
$$
and the Gaussian coefficients are defined by 
$$ 
\left[ \myatop{n}{i} \right]_q=\frac{(q^n-1)(q^{n-1}-1) \cdots (q^{n-i+1}-1)}{(q^i-1)(q^{i-1}-1) \cdots (q-1)}. 
$$
In addition, the number of blocks in this design is 
$$ 
b=q^{m-d} \left[ \myatop{m}{d} \right]_q. 
$$
In particular, $\AG_1(m, q)$ is a Steiner system\index{Steiner system} $S(2, q, q^m)$. 
When $d \geq 2$, $\AG_d(m, 2)$ is a $3$-design. In particular, $\AG_2(m, 2)$ is 
a Steiner system $S(3, 4, 2^m)$.    
\end{theorem}

To study the code of the design $\AG_r(m, q)$, we need to define a cyclic code. 
Let $q$ be a prime power as before. For any integer $j=\sum_{i=0}^{m-1}j_iq^i$, where $0 \leq j_i 
\leq q-1$ for all $0 \leq i \leq m-1$ and $m$ is a positive integer, we define 
\begin{eqnarray}\label{eqn-qweight}
\wt_q(j)=\sum_{i=0}^{m-1} j_i, 
\end{eqnarray}
where the sum is taken over the ring of integers, and is called the $q$-weight of $j$. 

Let $t \geq 0$ be an integer with $t=a(q-1)+b \leq m(q-1)$, where $0 \leq b \leq q-1$. 
We define a cyclic code $\M^t$ over $\gf(q)$ with length $q^m-1$ and defining set 
$$ 
\{i: 1 \leq i \leq q^m-1, \ \wt_q(i) < t\}. 
$$ 
Let $\overline{\M^t}$ denote the extended code of $\M^t$. The following theorem in the case 
that $q$ is a prime was proved in \cite{AK98}. It is also true for $q$ being any prime power.

\begin{theorem} \cite{AK98} 
Let $0 \leq r \leq m$. 
The code $\overline{\M^t}$ over $\gf(q)$ has length $q^m$, dimension 
$$ 
|\{i: 0 \leq i \leq q^m-1, \ \wt_q(i) \leq m(q-1)-t \}| 
$$ 
and minimum weight $(b+1)q^a$, where $t=a(q-1)+b$, $0 \leq a \leq m-1$, $0 \leq b <q-1$ 
and $(a, b) \neq (0,0)$.  
\end{theorem} 

The next result will be used later. 

\begin{theorem}\label{thm-Feb71}   \cite{AK98}
Let $0 \leq r \leq m$. 
The code $\C_{q}(\AG_r(m, q))$ of the design $\AG_r(m, q)$ of points and $r$-flats 
of the affine geometry $\AG(m, q)$ is the code $\overline{\M^{r(q-1)}}$ with minimum 
weight $q^r$ and dimension 
$$ 
|\{i: 0 \leq i \leq q^m-1, \ \wt_q(i) \leq (m-r)(q-1) \}|.  
$$  
\end{theorem} 

As corollaries of Theorem \ref{thm-Feb71}, we have the next two results.

\begin{corollary}\label{cor-feb121}  \cite{Hamada68,Hamada73} 
The code $\C_{q}(\AG_{m-1}(m, q))$ of the geometric design $\AG_{m-1}(m, q)$ of points and $(m-1)$-flats 
of the affine geometry $\AG(m, q)$ has length $q^m$, minimum weight $q^{m-1}$ and dimension 
$\binom{m+p-1}{m}^s$, where $q=p^s$. 
\end{corollary}

\begin{corollary}   \cite{Hamada68,Hamada73} 
The code $\C_{q}(\AG_{1}(m, q))$ of the geometric design $\AG_{1}(m, q)$ of points and lines  
of the affine geometry $\AG(m, q)$ has length $q^m$, minimum weight $q$. The dimension of the 
code is $q^m-\binom{m+q-2}{m}$ if $q$ is a prime. 

In particular, the code $\C_{3}(\AG_{1}(m, 3))$ of the Steiner triple system of 
points and lines of $\AG(m, 3)$ has parameters $[3^m, 3^m-1-m, 3]$. 
\end{corollary} 

\subsection{Linear codes from the $t$-designs held in the generalised Reed-Muller codes}\label{sec-june262b}

Let $\ell$ be a positive integer with $1 \leq \ell <(q-1)m$. The $\ell$-th order 
\emph{punctured generalized Reed-Muller code}\index{punctured generalized Reed-Muller code} 
$\cR_q(\ell, m)^*$ over $\gf(q)$ is the cyclic code of length $n=q^m-1$ with generator polynomial 
\begin{eqnarray}\label{eqn-generatorpolyPGRMcode}
g(x) = \sum_{\myatop{1 \leq j \leq n-1}{ \wt_q(j) < (q-1)m-\ell}} (x - \alpha^j), 
\end{eqnarray}
where $\alpha$ is a generator of $\gf(q^m)^*$. Since $\wt_q(j)$ is a constant function on 
each $q$-cyclotomic coset modulo $n=q^m-1$, $g(x)$ is a polynomial over $\gf(q)$. 

The parameters of the punctured generalized Reed-Muller code $\cR_q(\ell, m)^*$ are known 
and summarized in the next theorem. 

\begin{theorem}\label{thm-GPRMcode}  \cite{AK98} 
For any $\ell$ with $0 \leq \ell <(q-1)m$, $\cR_q(\ell, m)^*$ is a cyclic code over $\gf(q)$ 
with length $n=q^m-1$, dimension 
$$ 
\kappa=\sum_{i=0}^\ell \sum_{j=0}^{m} (-1)^j \binom{m}{j} \binom{i-jq+m-1}{i-jq} 
$$
and minimum weight $d=(q-\ell_0)q^{m-\ell_1-1}-1$, where $\ell=\ell_1(q-1)+\ell_0$ and 
$0 \leq \ell_0 < q-1$. 
\end{theorem} 

The following is also well known in the literature and will be needed later. 

\begin{theorem}\label{thm-puncgrm1}  \cite{Dingbk18} 
It is also know that 
$\cR_q(1, m)^*$ has parameters $[q^m-1, m+1, (q-1)q^{m-1}-1]$ and weight enumerator 
$$ 
1+(q-1)(q^m-1)z^{(q-1)q^{m-1}-1} + (q^m-1)z^{(q-1)q^{m-1}} + (q-1)z^{q^m-1}. 
$$ 
The dual code $(\cR_q(1, m)^*)^\perp$ has parameters $[q^m-1, q^m-m-2, d^\perp]$, where 
$d^\perp = 4$ if $q=2$, and $d^\perp = 3$ if $q \geq 3$. 
\end{theorem}

%\begin{theorem}\label{thm-dualPGRMcode}
For $0 \leq \ell < m(q-1)$, the code $(\cR_q(\ell, m)^*)^\perp$ is the cyclic code with generator polynomial 
\begin{eqnarray}\label{eqn-dualPGRMcodegenerator}
g^\perp(x) = \sum_{\myatop{0 \leq j \leq n-1}{ \wt_q(j) < \ell}} (x - \alpha^j), 
\end{eqnarray}
where $\alpha$ is a generator of $\gf(q^m)^*$. In addition, 
$$ 
(\cR_q(\ell, m)^*)^\perp = (\gf(q)\bone)^\perp \cap \cR_q(m(q-1)-1-\ell, m)^*, 
$$
where $\bone$ is the all-one vector in $\gf(q)^n$ and $\gf(q)\bone$ denotes the code over $\gf(q)$ with 
length $n$ generated by $\bone$.  
%\end{theorem}

The parameters of the dual of the punctured generalized Reed-Muller code are summarized as follows \cite{AK92}. 
%\begin{corollary} 
For $0 \leq \ell < m(q-1)$, the code $(\cR_q(\ell, m)^*)^\perp$ has length $n=q^m-1$, dimension 
$$ 
\kappa=n-\sum_{i=0}^\ell \sum_{j=0}^{m} (-1)^j \binom{m}{j} \binom{i-jq+m-1}{i-jq},  
$$
and minimum weight 
\begin{eqnarray}\label{eqn-lbmwtdualGPRMcode}
d \geq (q-\ell'_0)q^{m-\ell'_1-1},
\end{eqnarray}
where $m(q-1)-1-\ell=\ell'_1 (q-1)+\ell'_0$ and 
$0 \leq \ell'_0 < q-1$. 
%\end{corollary} 

The generalized Reed-Muller code $\cR_q(\ell, m)$ is defined to be the extended code 
of $\cR_q(\ell, m)^*$, and its parameters are given below \cite{AK98}. 
%\begin{theorem}\label{thm-GRMcodePara}
Let $0 \leq \ell < q(m-1)$. Then the generalized Reed-Muller code $\cR_q(\ell, m)$ has length $n=q^m$, dimension 
$$ 
\kappa=\sum_{i=0}^\ell \sum_{j=0}^{m} (-1)^j \binom{m}{j} \binom{i-jq+m-1}{i-jq},  
$$
and minimum weight 
\begin{eqnarray*}
d = (q-\ell_0)q^{m-\ell_1-1},
\end{eqnarray*}
where $\ell=\ell_1 (q-1)+\ell_0$ and 
$0 \leq \ell_0 < q-1$.  
%\end{theorem} 

The following is a well known result \cite{AK98} and will be needed shortly. 

\begin{theorem}\label{thm-DGM70}
Let $0 \leq \ell < q(m-1)$ and $\ell=\ell_1 (q-1)+\ell_0$, where $0 \leq \ell_0 < q-1$.  
The total number $A_{(q-\ell_0)q^{m-\ell_1-1}}$ of minimum weight codewords in $\cR_q(\ell, m)$ 
is given by 
\begin{eqnarray*}
A_{(q-\ell_0)q^{m-\ell_1-1}} =
(q-1) \frac{q^{\ell_1} (q^{m}-1) (q^{m-1}-1) \cdots (q^{\ell_1+1}-1)}{(q^{m-\ell_1}-1)(q^{m-\ell_1-1}-1) \cdots (q-1)}N_{\ell_0}, 
\end{eqnarray*} 
where 
\begin{eqnarray*}
N_{\ell_0} = \left\{ 
\begin{array}{ll}
1   & \mbox{ if } \ell_0=0, \\ 
\binom{q}{\ell_0} \frac{q^{m-\ell_1}-1}{q-1} & \mbox{ if } 0 < \ell_0 < q-1.  
\end{array} 
\right. 
\end{eqnarray*}
\end{theorem}

The generalized Reed-Muller codes $\cR_q(\ell, m)$ can also be defined with a multivariate polynomial approach. 
The reader is referred to \cite[Section 5.4]{AK98} for details. For $\ell < (q-1)m$, it was shown in \cite{AK98} that  
$$ 
\cR_q(\ell, m)^\perp = \cR_q(m(q-1)-1-\ell, m). 
$$

The general affine group $\GA_1(\gf(q))$ is defined by 
$$ 
\GA_1(\gf(q))=\{ax+b: a \in \gf(q)^*, \ b \in \gf(q) \}, 
$$ 
which acts on $\gf(q)$ doubly transitively \cite[Section 1.7]{Dingbk18}.  A linear code $\C$ 
of length $q$ is said to be affine-invariant if $\GA_1(\gf(q))$ fixes $\C$ \cite{Charpin90}. For affine-invariant 
codes we use the elements of $\gf(q)$ to index the coordinates of their codewords. 

Let $\ell$ be a positive integer with $1 \leq \ell <(q-1)m$, and let $q$ be a prime. Then $\cR_q(\ell, m)$ 
is affine-invariant, and the automorphism group $\Aut(\cR_q(\ell, m))$ is doubly transitive. These are well 
known facts about the generalized Reed-Muller codes $\cR_q(\ell, m)$. The results in the next two theorems 
are also well known (see \cite{AK92} or \cite{Dingbk18}) and follow from Theorems \ref{thm-designCodeAutm} 
and \ref{thm-DGM70}.

\begin{theorem}\label{thm-gRMcodeExt2Desi} 
Let $\ell$ be a positive integer with $1 \leq \ell <(q-1)m$. 
Then the supports of the codewords of weight $i >0$ in 
$\cR_q(\ell, m)$ form a $2$-design, provided that $A_i \neq 0$. 
\end{theorem} 

\begin{theorem}\label{thm-DGM70design}
Let $0 \leq \ell < q(m-1)$ and $\ell=\ell_1 (q-1)+\ell_0$, where $0 \leq \ell_0 < q-1$. 
The supports of minimum weight codewords in $\cR_q(\ell, m)$ form a $2$-$(q^m, (q-\ell_0)q^{m-\ell_1-1}, \lambda)$ design, where 
$$ 
\lambda= \frac{A_{(q-\ell_0)q^{m-\ell_1-1}}}{q-1} \frac{\binom{(q-\ell_0)q^{m-\ell_1-1}}{2}}{\binom{q^m}{2}}
$$ 
and $A_{(q-\ell_0)q^{m-\ell_1-1}}$ was given in Theorem \ref{thm-DGM70}. 
\end{theorem}

Note that $\cR_q(\ell, m)$ does not hold $3$-designs when $q>2$. 
It is known that 
$\cR_q(1, m)$ has parameters $[q^m, 1+m, (q-1)q^{m-1}]$ and weight enumerator  
\begin{eqnarray}\label{eqn-wtenumerator1stRMcode}
1+q(q^m-1)z^{(q-1)q^{m-1}} +(q-1)z^{q^m}.
\end{eqnarray} 
Furthermore, the supports of all minimum weight codewords in $\cR_q(1, m)$ form a 
$2$-$(q^m, (q-1)q^{m-1}, (q-1)q^{m-1}-1)$ design \cite{Dingbk18}.

We are now ready to present another result of this paper in the following theorem.

\begin{theorem}\label{thm-major1st} 
Let $\bD_{(q-1)q^{m-1}}(\cR_q(1, m))$ denote the $2$-design formed by the codewords of weight $(q-1)q^{m-1}$ in 
$\cR_q(1, m)$. Then $\C_p(\bD_{(q-1)q^{m-1}}(\cR_q(1, m)))$ has parameters 
$$ 
\left[q^m, \ \binom{p+m-1}{m}^s, \ q^{m-1}   \right],  
$$ 
where $q=p^s$. 
\end{theorem}  

\begin{proof}
Note that each codeword of weight $(q-1)q^{m-1}$ in $\cR_q(1, m)$ can be written as 
$$ 
\bc_{(a,b)}=(\tr(ax)+b)_{x \in \gf(q^m)}, \ a \in \gf(q^m)^*, \ b \in \gf(q).  
$$ 
We index the coordinates of the code $\cR_q(1, m)$ with the elements of $\gf(q^m)$. 
Then the support of the codeword $\bc_{(a,b)}$ is given by 
$$ 
\support(\bc_{(a,b)})=\{x \in \gf(q^m): \tr(ax)+b \neq 0\}. 
$$ 
The complement of $\support(\bc_{(a,b)})$ with respect to $\gf(q^m)$ is given by 
$$ 
\support(\bc_{(a,b)})^c=\{x \in \gf(q^m): \tr(ax)+b = 0\},  
$$ 
which is an $(m-1)$-flat in $\gf(q^m)$ when $\gf(q^m)$ is viewed as an $m$-dimensional 
vector space over $\gf(q)$. Consequently, the complementary design $\bD_{(q-1)q^{m-1}}(\cR_q(1, m))^c$ of $\bD_{(q-1)q^{m-1}}(\cR_q(1, m))$ is 
the design $\AG_{m-1}(m, q)$ of points and $(m-1)$-flats of the affine geometry $\AG(m, q)$. 

We now prove that the all-one vector $\bone$ is a codeword in both $\C_p(\bD_{(q-1)q^{m-1}}(\cR_q(1, m)))$ and 
$\C_p(\AG_{m-1}(m, q))$. 
It is well known that the design $\AG_{m-1}(m, q)$ has parameters 
$$ 
2\mbox{ --}\left(q^m, \ q^{m-1}, \ \frac{q^{m-1}-1}{q-1}   \right).  
$$  
Therefore each point is incident with the following number of blocks: 
$$ 
\lambda^c_1= \frac{q^{m-1}-1}{q-1} \frac{\binom{q^m-1}{1}}{\binom{q^{m-1}-1}{1}}=\frac{q^m-1}{q-1}. 
$$ 
It then follows that the sum over $\gf(q)$ of the row vectors of the incidence matrix of 
the design $\AG_{m-1}(m, q)$ is 
$$ 
(\lambda^c_1, \lambda^c_1, ..., \lambda^c_1)=(1, 1, \cdots, 1)=\bone, 
$$ 
which is a codeword in $\C_p(\AG_{m-1}(m, q))$. 

Since $\bD_{(q-1)q^{m-1}}(\cR_q(1, m))$ is a $2$-$(q^m, q^m-q^{m-1}, (q-1)q^{m-1}-1)$ design, every point of 
the design $\bD_{(q-1)q^{m-1}}(\cR_q(1, m))$ is incident 
with the following number of blocks: 
$$ 
\lambda_1=\left((q-1)q^{m-1}-1\right) \frac{\binom{q^m-1}{1}}{\binom{q^m-q^{m-1}-1}{1}}
=q^m-1 
$$ 
We then deduce that the sum over $\gf(q)$ of the row vectors of the incidence matrix of 
the design 
$\bD_{(q-1)q^{m-1}}(\cR_q(1, m))$ is 
$$ 
(\lambda_1, \lambda_1, ..., \lambda_1)=(-1, -1, \cdots, -1)=-\bone, 
$$ 
which is a codeword in $\C_p(\bD_{(q-1)q^{m-1}}(\cR_q(1, m)))$. Consequently, $\bone \in \C_p(\bD_{(q-1)q^{m-1}}(\cR_q(1, m)))$. 

It then follows from Theorem \ref{thm-nyear19} that $\C_p(\bD_{(q-1)q^{m-1}}(\cR_q(1, m)))$ is equal to $\C_p(\AG_{m-1}(m, q))$. The desired conclusion then follows from Corollary 
\ref{cor-feb121}.  
\end{proof}

When $q=2$, it is easily seen that $\C_p(\bD_{(q-1)q^{m-1}}(\cR_q(1, m)))$ equals $\cR_q(1, m)$. However, the two 
codes are very different if $q>2$. This is obvious from the dimensions of the two codes. 
Note that the design $\bD_{(q-1)q^{m-1}}(\cR_q(1, m))$ in Theorem \ref{thm-major1st} is not a geometric design. But 
its code over $\gf(q)$ is the same as the code over $\gf(q)$ of the geometric design 
$\AG_{m-1}(m, q)$. Our contribution is mainly to prove this fact.

\begin{theorem}\label{thm-major1stdual} 
Let $\bD_{(q-1)q^{m-1}}(\cR_q(1, m))$ denote the $2$-design formed by the codewords of weight $(q-1)q^{m-1}$ in 
$\cR_q(1, m)$. Then $\C_p(\bD_{(q-1)q^{m-1}}(\cR_q(1, m)))^\perp$ has parameters 
$$ 
\left[q^m, \ q^m-\binom{p+m-1}{m}^s, \ d^\perp   \right],  
$$ 
where $q=p^s$, $d^\perp \geq q+2$ if $s>1$ and $d^\perp=2p$ if $s=1$. 
\end{theorem}  

\begin{proof} 
The dimension of the code $\C_p(\bD_{(q-1)q^{m-1}}(\cR_q(1, m)))^\perp$ follows from Theorem \ref{thm-major1st}. 
Recall that $\C_p(\bD_{(q-1)q^{m-1}}(\cR_q(1, m)))=\C_p(\AG_{m-1}(m, q))$. It then follows from Theorem \ref{thm-AKp54} 
that $d^\perp \geq q+2$. If $s=1$, it then follows from Theorem 5.7.9 in \cite{AK98} that 
$d^\perp=2p$.   
\end{proof}

Notice that $\cR^*_q(1, m)$ is a cyclic code and invariant under the general linear group 
$\GL_m(q)$, which is transitive on $\gf(q^m)^*$. By Theorem \ref{thm-puncgrm1}, $\cR^*_q(1, m))$ 
is a three-weight code. Hence, $\cR^*_q(1, m))$ holds two $1$-designs. One of them is 
$\bD_{(q-1)q^{m-1}-1}(\cR^*_q(1, m))$ with parameters 
$$ 
1 \mbox{--} \left(q^m-1, \ (q-1)q^{m-1}-1, \ (q-1)q^{m-1}-1 \right). 
$$  
The other is the design $\bD_{(q-1)q^{m-1}}(\cR^*_q(1, m))$ with parameters 
$$ 
1 \mbox{--} \left(q^m-1, \ (q-1)q^{m-1}, \ q^{m-1} \right). 
$$   
By definition, $\C_p(\bD_{(q-1)q^{m-1}-1}(\cR^*_q(1, m)))$ is a punctured code of 
$\C_p(\bD_{(q-1)q^{m-1}}(\cR_q(1, m)))$. The following result then follows from 
Theorem \ref{thm-major1st}.     

\begin{theorem}\label{thm-major1st2} 
Let $\bD_{(q-1)q^{m-1}-1}(\cR^*_q(1, m))$ denote the $1$-design formed by the codewords of weight $(q-1)q^{m-1}-1$ in 
$\cR^*_q(1, m)$. Then $\C_p(\bD_{(q-1)q^{m-1}-1}(\cR^*_q(1, m)))$ has parameters 
$$ 
\left[q^m-1, \ \binom{p+m-1}{m}^s, \ q^{m-1}-1   \right],  
$$ 
where $q=p^s$. 
\end{theorem}

\begin{theorem}\label{thm-major1st3} 
Let $\bD_{(q-1)q^{m-1}}(\cR^*_q(1, m))$ denote the $1$-design formed by the codewords of weight $(q-1)q^{m-1}$ in 
$\cR^*_q(1, m)$. Then $\C_p(\bD_{(q-1)q^{m-1}}(\cR^*_q(1, m)))$ has parameters 
$$ 
\left[q^m-1, \ \binom{p+m-2}{m-1}^s, \ d   \right],  
$$ 
where $q=p^s$, $d=(q-1)d(\C_p(\bD))\geq q^{m-1}-1$, $\C_p(\bD)$ is the code of Theorem \ref{thm-majorthm2}, and $d(\C_p(\bD))$ denotes the minimum distance of the code $\C_p(\bD)$. 
\end{theorem}  

\begin{proof}
It is straightforward to see that $\bD_{(q-1)q^{m-1}}(\cR^*_q(1, m))$ is the design held 
by the supports of codewords of weight $(q-1)q^{m-1}$ in the code 
$$ 
\C=\{(\tr_{q^m/q}(ax))_{x \in \gf(q^m)^*}: a \in \gf(q^m)\}, 
$$ 
which is equivalent to a concatenation of $q-1$ copies the first-order projective Reed-Muller 
code. The desired conclusions then follow from Theorem \ref{thm-majorthm2}. 
\end{proof}

Once we determine the minimum weight of the code $\C_p(\bD)$ in Theorem \ref{thm-majorthm2}, 
we will be able to determine the minimum weight of $\C_p(\bD_{(q-1)q^{m-1}}(\cR^*_q(1, m)))$, 
and vice versa.

The following problem is very hard to settle. But we will solve it for a few special cases in the rest of this section. 

\begin{open}\label{open-june261} 
Determine the parameters of $\C_p(\bD_{i}(\cR_q(\ell, m)))$ for other designs $\bD_{i}(\cR_q(\ell, m))$ held in 
$\cR_q(\ell, m)$ for $\ell \geq 2$, and study properties of $\C_p(\bD_{i}(\cR_q(\ell, m)))$.  
\end{open} 

The parameters of the designs held in $\cR_q(\ell, m)$ are still open. Even the weight 
distribution of the code $\cR_q(\ell, m)$ is open for $\ell \geq 3$ and $q>2$. The 
weight distribution of $\cR_q(2, m)$ is known for $q>2$ \cite{Li2019}. It may be possible 
to settle the parameters of $\C_p(\bD_{i}(\cR_q(2, m)))$ for $q>2$ and some $i$. 

\begin{table}[htbp]
\centering
\caption{The parameters of $\cR_p(r, m)$ and $\C_p(\bD_{d}(\cR_p(r, m)))$}
\label{table-july131}
\begin{tabular}{|c|c|c|}
  \hline
$(p,m,r)$ & $\cR_p(r, m)$  &  $\C_p(\bD_{d}(\cR_p(r, m)))$ \\  \hline
$(3,2,1)$ & $[9,3,6]$    & $[9,6,3]$  \\ \hline 
$(3,3,1)$ & $[27,4,18]$    & $[27,10,9]$  \\ \hline 
$(3,4,1)$ & $[81,5,54]$    & $[81,15,27]$  \\ \hline 
$(3,3,2)$ & $[27,10,9]$    & $[27,10,9]$  \\ \hline 
$(3,4,2)$ & $[81,15,27]$    & $[81,15,27]$  \\ \hline 
$(5,2,2)$ & $[25,6,15]$    & $[25,15,5]$  \\ \hline 
$(3,3,3)$ & $[27,17,6]$    & $[27,23,3]$  \\ \hline 
\end{tabular}
\end{table}

A comparison between the parameters of $\cR_p(r, m)$ and $\C_p(\bD_{d}(\cR_p(r, m)))$ 
is given in Table \ref{table-july131}, where $d$ is the minimum distance of $\cR_p(r, m)$. 
In general the parameters of the two codes $\cR_p(r, m)$ and $\C_p(\bD_{d}(\cR_p(r, m)))$ 
are different. However, in the special case $(p, r)=(3,2)$ we have the following.   

\begin{theorem}\label{thm-july15} 
For $m \geq 2$ the two codes $\cR_3(2, m)$ and $\C_3(\bD_{3^{m-1}}(\cR_3(2, m)))$ are identical.
\end{theorem}

\begin{proof}
By Theorems \ref{thm-GPRMcode} and \ref{thm-DGM70}, $\cR_3(2, m)$ has minimum distance 
$d=3^{m-1}$ and dimension 
$
k=1+m+\frac{(m+1)m}{2}. 
$   
In addition,  the total number of minimum weight codewords in  $\cR_3(2, m)$ is 
$  
A_d=3(3^m-1).  
$  

Let $\tr$ denote the trace function from $\gf(3^m)$ to $\gf(3)$. 
It is easily seen that the set of all minimum weight codewords in $\cR_3(2,m)$
is given by
$$\left \{ \pm \left ( (\tr(ax)+b)^2 - 1 \right )_{x\in \mathrm{GF}(q)}: (a,b)\in \mathrm{GF}(3^m) \times \mathrm{GF}(3)^* \right \}.$$
Since $ (\tr(ax)+b)^2 - 1=0 \text{ or } -1$, the code $\C_3(\bD_{3^{m-1}}(\cR_3(2, m)))$ is linearly spanned by the codewords
 in the following set:
 \begin{align}\label{eq-hp2}
 \left \{  \left ( (\tr(ax)+b)^2 - 1 \right )_{x\in \mathrm{GF}(q)}: (a,b)\in \mathrm{GF}(3^m) \times \mathrm{GF}(3)^* \right \}.
 \end{align}
Let $b_1, b_2 \in \mathrm{GF}(3^m)$. By (\ref{eq-hp2}), we have
$$\left (  \tr((b_1+b_2)x)^2- \tr((b_1-b_2)x)^2 \right )_{x\in \mathrm{GF}(3^m)} \in \C_3(\bD_{3^{m-1}}(\cR_3(2, m))),$$
which is the same as
\begin{align}\label{eq-b1x-b2x}
\left (  \tr(b_1x)\tr(b_2x) \right )_{x\in \mathrm{GF}(3^m)}  \in \C_3(\bD_{3^{m-1}}(\cR_3(2, m))),
\end{align}
for all $b_1, b_2 \in \mathrm{GF}(3^m)$.
Let $b\in \mathrm{GF}(3^m)^*$.
By (\ref{eq-hp2}), $\left ( (\tr(bx)-1)^2 - 1 \right )_{x\in \mathrm{GF}(q)} \in \C_3(\bD_{3^{m-1}}(\cR_3(2, m)))$.
By (\ref{eq-b1x-b2x}), $\left ( \tr(bx)^2 \right )_{x\in \mathrm{GF}(q)} \in \C_3(\bD_{3^{m-1}}(\cR_3(2, m)))$.
Note that $\tr(bx)=( (\tr(bx)-1)^2 - 1)-  \tr(bx)^2 $. Then
\begin{align}\label{eq-tr(bx)}
\left ( \tr(bx) \right )_{x\in \mathrm{GF}(q)} \in \C_3(\bD_{3^{m-1}}(\cR_3(2, m))).
\end{align}
By (\ref{eq-hp2}) and (\ref{eq-b1x-b2x}), we have $\left ( 1 \right )_{x\in \mathrm{GF}(q)} \in \C_3(\bD_{3^{m-1}}(\cR_3(2, m))).$
Thus, $\C_3(\bD_{3^{m-1}}(\cR_3(2, m)))$ is linearly spanned by the set
$$\left \{ \left ( \tr(ax)\tr(bx) \right )_{x\in \mathrm{GF}(q)}, \left ( \tr(ax) \right )_{x\in \mathrm{GF}(q)},
 \left (1 \right )_{x\in \mathrm{GF}(q)}: a,b\in \mathrm{GF}(3^m)\right \}.$$
 It is observed that 
 the linear space spanned by $$\left \{ \left ( \tr(ax)\tr(bx) \right )_{x\in \mathrm{GF}(q)}, \left ( \tr(ax) \right )_{x\in \mathrm{GF}(q)},
 \left (1 \right )_{x\in \mathrm{GF}(q)}: a,b\in \mathrm{GF}(3^m)\right \}$$ is exactly 
$\cR_3(2,m)$.
This completes the proof. 
\end{proof}

%\begin{color}{blue}
%Our last task of this paper is the following. 

%\begin{task} 
%Study the code $\C_3(\bD_{3^{m-1}}(\cR_3(3, m)))$. 
%\end{task} 

%This is to deal with the third-order Reed-Muller codes. It may be feasible, as we deal with the case $p=3$. 
%\end{color} 

\section{Summary and concluding remarks} 

Using the results on the linear codes of geometric designs and the generalised Reed-Muller codes 
documented in \cite{AK92},  this paper made the following contributions: 
\begin{itemize}
\item The results about $\C_2(\bD_i(\C))$, $\C$, $\bD_i(\C)$, and their automorphism groups for binary linear codes $\C$ documented  in Section \ref{sec-binarycase}. 

\item The determination of some of the parameters of the linear code $\C_p(\bD)$ documented in Theorem \ref{thm-majorthm2}, 
          where $\bD$ is the design held in a code related to the first-order projective  Reed-Muller code.   
          
\item The determination of the parameters of the linear code  $\C_p(\bD_{(q-1)q^{m-1}}(\cR_q(1, m)))$ documented 
          in Theorem \ref{thm-major1st} and its dual $\C_p(\bD_{(q-1)q^{m-1}}(\cR_q(1, m)))^\perp$ documented in 
          Theorem \ref{thm-major1stdual}, where $\bD_{(q-1)q^{m-1}}(\cR_q(1, m))$ is the design supported by the 
          codewords of Hamming weight $(q-1)q^{m-1}$ in the Reed-Muller code $\cR_q(1, m)$.        
          
\item The determination of the parameters of the ternary code $\C_3(\bD_{3^{m-1}}(\cR_3(2, m)))$ documented in 
          Theorem \ref{thm-july15}.           

\item The determination of the parameters of the linear code  $\C_p(\bD_{(q-1)q^{m-1}-1}(\cR^*_q(1, m)))$ documented 
          in Theorem \ref{thm-major1st2}, where $\bD_{(q-1)q^{m-1}-1}(\cR^*_q(1, m))$ is the design supported by the 
          codewords of Hamming weight $(q-1)q^{m-1}-1$ in the punctured generalised Reed-Muller code $\cR^*_q(1, m)$.  
          
\item The determination of the parameters of the linear code  $\C_p(\bD_{(q-1)q^{m-1}}(\cR^*_q(1, m)))$ documented 
          in Theorem \ref{thm-major1st3}, where $\bD_{(q-1)q^{m-1}}(\cR^*_q(1, m))$ is the design supported by the 
          codewords of Hamming weight $(q-1)q^{m-1}$ in the punctured generalised Reed-Muller code $\cR^*_q(1, m)$.   
\end{itemize} 
These summarize the new results presented in this paper. 

Although the designs considered in this paper are not geometric designs and the linear codes are not geometric codes and Reed-Muller codes, they are closely related to geometric designs and the Reed-Muller codes. Thus, in Sections \ref{sec-june261a} and \ref{sec-june262a} we had to introduce these geometric codes and the Reed-Muller 
codes  as well as their basic properties. This took quite some space.   

As observed, it is extremely hard to get information on the code $\C_p(\bD_i(\C))$ for general linear codes over $\gf(q)$ 
for nonbinary codes $\C$ holding designs. The reader is cordially invited to settle Conjectures \ref{conj-june261} and 
\ref{conj-june262} and Open Problem \ref{open-june261}. The rank of $t$-designs, i.e., the dimension of the corresponding codes, 
may be used to classify $t$-designs of certain type. For example, the rank of Steiner triples was intensively studied and 
employed for classifying Steiner triple systems \cite{JMTW}. 

Finally, we point out that the idea of using a linear code $\C_1$ supporting a $t$-design 
$\bD_w(\C)$ to obtain a new linear code $\C_{q}(\bD_w(\C))$ may produce a bad or good code. 
Distance-optimal ternary linear codes were obtained in \cite{DTT19} with this method.

\section*{Acknowledgements} 

C. Ding's research was supported by the Hong Kong Research Grants Council, Proj. No. 16300418. C. Tang was supported by National Natural Science Foundation of China (Grant No.
11871058) and China West Normal University (14E013, CXTD2014-4 and the Meritocracy Research
Funds)

\end{document}